\documentclass[11pt]{article} 
\usepackage{kotex} 
\usepackage{graphicx}
\usepackage{epsfig} 
\usepackage{amssymb, amsmath}
\usepackage{verbatim} 
\usepackage{caption}
\usepackage[noblocks]{authblk} 
\usepackage{multirow}
\usepackage{subfig}

\usepackage[usenames,dvipsnames]{color}

\newtheorem{theorem}{Theorem}[section]
\newtheorem{lemma}[theorem]{Lemma}

\newenvironment{proof}[1][Proof]{\begin{trivlist}
\item[\hskip \labelsep {\bfseries #1}]}{\end{trivlist}}

\newcommand{\E}{{\rm E}}
\newcommand{\V}{{\rm Var}}

\newcommand{\bbR}{\mathbb{R}}

\newcommand{\half}{\frac{1}{2}}

\newcommand{\bea}{\begin{eqnarray*}}
\newcommand{\eea}{\end{eqnarray*}}
\newcommand{\bean}{\begin{eqnarray}}
\newcommand{\eean}{\end{eqnarray}}
\newcommand{\baln}{\begin{align}}
\newcommand{\ealn}{\end{align}}

\begin{document}

\title{Inference for Differential Equation Models using Relaxation via Dynamical Systems}
\author[1]{Kyoungjae Lee}
\author[2]{Jaeyong Lee}
\author[3]{Sarat Dass}
\affil[1]{Department of Applied and Computational Mathematics and Statistics, The University of Notre Dame}
\affil[2]{Department of Statistics, Seoul National University}
\affil[3]{Department of Fundamental and Applied Sciences, Universiti Teknologi PETRONAS}

\maketitle

\begin{abstract}
Statistical regression models whose mean functions are represented by ordinary differential equations (ODEs) can be used to describe phenomenons  dynamical in nature, which are abundant in areas such as biology, climatology and genetics.
The estimation of parameters of ODE based models is essential for understanding its dynamics,  but the lack of an analytical solution of the ODE makes  the parameter estimation challenging. 
	The aim of this paper is to propose a general and fast framework of statistical inference for ODE based models by relaxation of the underlying ODE system. Relaxation is achieved by a properly chosen numerical procedure, such as the Runge-Kutta, and by introducing additive Gaussian noises  with small variances.  Consequently, filtering methods can be applied to obtain the posterior distribution of the parameters in the Bayesian framework.  The main advantage of the proposed method is computation speed. In a simulation study, the proposed method was at least 14 times faster than the other methods. 
Theoretical results which guarantee the convergence of the posterior of the approximated dynamical system to the posterior of true model  are  presented. Explicit expressions are given that relate the order and the mesh size of the Runge-Kutta procedure to  the rate of convergence of the approximated posterior as a function of sample size.

\end{abstract}

\section{Introduction}
Many dynamical phenomenons in the real world can be represented mathematically by ordinary differential equations (ODEs). Common examples include Newton's law of cooling, Lotka-Volterra equations for predator-prey populations (Alligood et al., 1997)\nocite{Alligood97} and Lorenz equation for atmospheric convection (Lorenz, 1963)\nocite{lorenz63}.
There are many other popular examples describing physical, chemical and biological phenomenons using ODEs.
Although observing the data sets from an ODE systems is common, estimating the parameters of ODE models (ODEMs) can be challenging because of  lack of an analytical solution to ODE. Here, we give a brief review of  previous works on the ODEMs.

There are several frequentist methods in the literature for parameter estimation of ODEMs. Bard (1974)\nocite{BardYonathan74} used numerical integration to approximate the solution of ODEs and minimized the objective function based on a gradient method.
Varah (1982)\nocite{Varah82} suggested a two step estimation method using the cubic spline approximation. The two steps consist of estimation of the regression function and estimation of the parameters in the ODEM. Ramsay and Silverman (2005)\nocite{RamsaySilverman05} modified the first step of Varah by adding the roughness penalty function which measures the difference between the ODE and the mean function.
The parameter cascading method was proposed by Ramsay et al. (2007)\nocite{RamsayHooker07}. They grouped the parameters into the regression coefficients, structural parameters, and regularization parameters. The parameters in each group are estimated in turn in a cascading fashion.

Bayesian inference of ODEMs is more challenging because naive application of Markov Chain Monte Carlo (MCMC) methods  would require calculation of  the numerical solution of ODE whenever parameters are sampled from the proposal distribution.
Gelman et al. (1996)\nocite{Gelman96} and Huang et al. (2006)\nocite{HuangLiu06} proposed a Bayesian computation  method for parameter inference of pharmacokinetic models and the longitudinal HIV dynamic system, respectively.
Campbell (2007)\nocite{Campbell07} combined the parallel tempering (Geyer, 1991)\nocite{Geyer92} and collocation method (Ramsay et al., 2007) to get over the rough surface of the posterior, but this slows down the speed of  computations significantly. Arnold et al. (2013)\nocite{arnold2013linear} used particle filter framework for the inference of ODEMs with linear multistep methods for the numerical integration.  
Dass et al. (2017)\nocite{dass2017laplace} suggested  a Bayesian inference with Laplace approximation for a fast computation when the dimension of $\theta$ is moderate.

In this paper, we propose a Bayesian inference method for the ODEMs using a relaxation technique via dynamical systems and associated dynamic models. Relaxation is achieved by a properly chosen numerical procedure, such as the Runge-Kutta, and by introducing additive Gaussian noise variables with variance tending to zero. The variance of the additive noise variables works as a measure of fidelity to the original ODEM and by letting it tend to zero, we recover the original model. The relaxation introduces inefficiency of the inference, but we gain the speed of the computation in return. 

For a fast computation, a filtering method is applied for inferring posterior distributions of parameters in a Bayesian framework. The relaxation technique provides a dynamical system and model to which a fast inference tool based on sequential Monte Carlo can be applied to.  With these sequential methods, we do not need to calculate the whole path of the numerical solution for each realization of the new parameter. It reduces the computation time significantly compared to other standard Bayesian procedures and enables us  to deal with the ODEM   in reasonable computing time. In subsection \ref{subsection:LV}, to emphasize its fast computation the proposed method is compared with the other methods: the parameter cascading, the delayed rejection adaptive Metropolis algorithm  and the Bayesian inference with the Laplace approximation. In the simulation study, the proposed method is from 14 times to 78 times faster than other methods.

We also derive convergence results for the approximated posteriors under suitable regularity conditions. We present a guideline for the choice of the model parameters which give a reasonable relative error rate, and provide its theoretical basis. Theoretical results which guarantee the convergence of the posterior of the approximated dynamical system to the posterior of true model are presented. Explicit expressions  are given that relate the order and the mesh size of the Runge-Kutta procedure and  guarantee the rate of convergence of the approximated posterior to the true posterior.


The rest of the paper is organized as follows.
In section \ref{sec:odems}, we describe a differential equation model and its corresponding relaxed dynamic model counterpart as well as prior choices.
The method of posterior inference is described in section \ref{section:ELW}. Some theoretical support for the proposed method are given in section \ref{section:Conv}. In section \ref{section:examples}, we give two simulated data examples to  demonstrate  the speed and  performance  of the proposed method. A real data set, the Lynx-Hare data set, is analyzed in section \ref{sec:realdata}. The discussion is given in section \ref{sec:disc}. The proofs of theorems are given in the appendix.


\section{Ordinary Differential Equation Models and Nonlinear Dynamic Models}\label{sec:odems}
\subsection{Ordinary Differential Equation Models (ODEMs)}
The ODEM is the regression model with regression function $x(t)$ described by an ODE. The regression function $x(t)$ is the solution of the differential equation 
\bean\label{dotx}
\dot{x}(t) = f(x, u, t ; \theta),
\eean
where $f$ is a $p$-dimensional smooth function,  $u(t)$ is a deterministic input function, $\theta \in \Theta \subset \bbR^q$ is the unknown parameter, and $\dot{x}(t)$ denotes the first derivative of $x(t)$ with respect to time $t$. Since the input function $u(t)$ does not affect the general ideas of inference in this paper, it is not considered subsequently. The data are observed at $n$ points in the time interval $t \in [0, T] \subset \Bbb{R}$, given by $0\le t_1,t_2,\ldots, t_n \le T$. Thus,
\bea
y_i = x(t_i) + \epsilon_i,~~ i=1,\ldots,n,
\eea
where $y_i$ is a $p$-dimensional observation vector at time $t_i$,  the error $\epsilon_i$ is drawn independently from the multivariate normal distribution $N_p(0, \sigma^2I_p)$ with unknown $\sigma^2 > 0 $, and $x(t_i)\equiv x_i$ is the underlying regression function measured at time $t_i$.

The regression model is given by 
\bean \label{ODE}
\begin{split}
y_i &= x_i + \epsilon_i ,~~ i=1,\ldots,n, \\
\dot{x}(t) &= f(x, u, t; \theta)
\end{split}
\eean 
where $x_i = x(t_i)$. The covariate $x_i$ is determined by the initial value of $x$, $x_0 = x(0)$, and the parameter $\theta$. In the rest of the paper, we call the model \eqref{ODE} as the regression model or the true model.

In most cases, ODE (\ref{dotx}) does not have a closed form solution, so there is a need to approximate $x(t)$ numerically. We will use the Runge-Kutta method which is a standard numerical method for ODE. While there are many types of Runge-Kutta methods, we will only consider the 4th order method in this paper. However, our proposed method can be extended to the other approximation methods for ODE as well as other Runge-Kutta methods with different orders easily. Letting $h_{i+1} = t_{i+1} - t_i$, the form of 4th order Runge-Kutta approximation for \eqref{ODE} is as follows:
\bean\label{RK4}
x_{i+1} &\equiv& g(x_i, t_i ; \theta) = x_i + {1 \over 6}(k_{i1} + 2k_{i2} + 2k_{i3} + k_{i4}), ~~ i = 0, \ldots, n - 1 ,
\eean
where
\bea
k_{i1} &=& h_{i+1}f(x_i, t_i ; \theta) , \\ 
k_{i2} &=& h_{i+1}f(x_i + {1 \over 2}k_{i1} , t_i + {1\over2}h_{i+1} ;\theta), \\ 
k_{i3} &=& h_{i+1}f(x_i + {1\over2}k_{i2},t_i + {1\over2}h_{i+1} ;\theta), \\ 
k_{i4} &=& h_{i+1}f(x_i + k_{i3} , t_i + h_{i+1};\theta). 
\eea
In the above equation, all $x_i$'s indicate the approximated values. 
For more details, see Spijker (1996).\nocite{Spijker96}

With this approximation, we have the following model
\bean \label{ODEmodel}
\begin{split}
y_i &= x_i + \epsilon_i  , ~~ i = 1 , \ldots, n, \\
{x}_{i+1} &= g(x_{i} , t_{i} ; \theta) , ~~ i = 0 , \ldots, n-1.
\end{split}
\eean 
In the remainder of this paper, we call the model \eqref{ODEmodel} as a differential equation model (DEM).
Sometimes to obtain better approximation of $x_{i+1}$, we divide the interval $[t_{i-1}, t_i]$ into $m$ small  subintervals and apply the Runge-Kutta method for the subintervals. In this case, we will call the corresponding ODE model the $m$ step ODE model and $m$ the step size.

\subsection{Nonlinear Dynamic Models}
In practice, estimating the parameter from DEM can pose a significant computational  challenge if the ODE does not have an analytical solution. Dass et al. (2017) marginalized out $x_0$ using Laplace approximation and conducted grid sampling to get posterior samples of $\theta$.  Their method is fast and accurate when the dimension of $\theta$ is small; however, the methodology suffers from heavy computations when the dimension of $\theta$ is large. The computation time increases  exponentially as the dimension of $\theta$ increases due to the grid sampling. 
The griddy Gibbs sampler can be used on $\theta$, but practical problems such as dependencies and slow convergence may arise.



In this paper, in order to make posterior inference on $\theta$, we adopt a  nonlinear dynamic model relaxation of the DEM in (\ref{ODEmodel}) given in terms of the model below with unknown initial condition $x_0$:
\bean\label{NDmodel}
\begin{split}
y_i &= \tilde{x}_i + \epsilon_i  , ~~ i = 1 , \ldots, n ,\\
\tilde{x}_{i+1} &= g(\tilde{x}_{i}, t_{i} ; \theta) + \eta_{i} , ~~ i = 0, \ldots, n-1
\end{split}
\eean
where $\epsilon_i \overset{iid}\sim N(0, \sigma^2I_p)$ and $\eta_i \overset{iid}\sim N(0, u^2I_p)$ with $\sigma, u >0$. 
The error term $\eta_i$ reflects the fact that the approximation $g(x_{i},t_{i};\theta)$ of $x_{i+1}$ is made with uncertainty.
In the remainder of the paper, we call model \eqref{NDmodel} as  the approximate dynamic model obtained as a relaxation of the DEM in (\ref{ODEmodel}) via the relaxation parameter $u$.  The quantities $\tilde{x}_i$ in \eqref{NDmodel} are not the same as $x_{i}$ given in \eqref{ODEmodel} since the former are quantities that are observed with error whereas the latter are not. However, note that the two models \eqref{ODEmodel} and \eqref{NDmodel} become equivalent as the relaxation parameter $u \to 0$.

In the above model \eqref{NDmodel}, there are four unknown quantities, namely, $x_0,\theta, \lambda = 1/\sigma^2$ and $u$. 
The Bayesian approach proceeds by considering priors for these quantities. 
We do not consider a prior for the relaxation parameter $u$ since it is artificially introduced to control the quality of the approximation. We fix $u$ to be a small positive quantity in the subsequent numerical computations.
The priors on $x_0$ and $\lambda$ are taken as
\bean\label{priors}
\begin{split}
x_0 | \lambda &\sim N_p(\mu_{x_0}, c\lambda^{-1}I_p) \,\mbox{ and }\\
\lambda  &\sim \text{Gamma}(a_\lambda, b_\lambda) , \\
\end{split}
\eean
where $c>0$ and $\text{Gamma}(a, b)$ represents the Gamma distribution with mean $a/b$ and variance $a/b^2$. 
The prior for $\theta$, $\pi(\theta)$, is taken independently of the rest of the unknown quantities above.

\subsection{Sequential Monte Carlo}
Sequential Monte Carlo (SMC) is a simulation-based method for estimating the states and the parameters of the nonlinear dynamic model. The basic idea of SMC is using the importance samples to approximate posterior at each state and updating the samples sequentially through a proper kernel. There exists an extensive literature on SMC which includes sequential importance sampling (Handschin and Mayne, 1969\nocite{handschin1969monte}), bootstrap filter (Gordon et al., 1993\nocite{Gordon93}), auxiliary particle filter (Pitt and Shephard, 1999\nocite{Pitt99}), Rao-Blackwellised particle filter (Doucet et al., 2000)\nocite{doucet2000rao}, sequential Monte Carlo sampler (Del Moral et al., 2006)\nocite{del2006sequential}, Liu and West filter (Liu and West, 2001\nocite{Liu01}), particle learning (Carvalho et al., 2010\nocite{Carvalho10}), multilevel sequential Monte Carlo sampler (Beskos et al., 2016)\nocite{beskos2016multilevel}, to name just a few. For an extensive review of SMC, see Doucet et al. (2001)\nocite{doucet2001introduction}, Kantas et al. (2009)\nocite{kantas2009overview}, Lopes and Tsay (2011)\nocite{lopes2011particle} or S{\"a}rkk{\"a} (2013)\nocite{sarkka2013bayesian}.

The SMC has advantages over other alternative posterior computation methods such as Kalman filter, extended Kalman filter and Markov chain Monte Carlo (MCMC). 
The Kalman filter and the extended Kalman filter are applicable to the linear dynamic model, while the SMC can be applied to the nonlinear dynamic model as well. 
The SMC has advantages over MCMC. First, SMC methods are much faster than MCMC methods. Whenever the new parameter is propagated in each stage of SMC, we only calculate the next step of the numerical solution. 
Fast computation is the  biggest advantage of our method. Second, they are able to be implemented in an on-line learning scenario. When a new data point is observed,  SMC just need to update one step of the algorithm, while MCMC must implement the whole algorithm again to get the new posterior samples. 
Due to these advantages, we choose SMC for the posterior computation of the nonlinear dynamic model, which approximates the ODE model. 


\section{Posterior Computations for the Approximate Dynamic Model via  Sequential Monte Carlo}\label{section:ELW}

To obtain inference for $\theta$ based on the approximated dynamic model of (\ref{NDmodel}), we will use the extended Liu and West (ELW) filter to estimate parameters and states (Rios and Lopes, 2013).\nocite{Rios13}  We call the proposed method of computation  relaxed DEM with ELW filter (RDEM-ELW)  or simply RDEM. 
The ELW filter uses the idea of auxiliary particle filter to sample the states, and it divides the parameters into two sets, $\theta$ and $\gamma$, representing parameters with and without sufficient statistic, respectively. The parameters denoted by $\theta$ (i.e., without the sufficient statistic) is the same set of parameters denoted by $\theta$ in (\ref{NDmodel}).  
For the $\theta$-set, the ELW filter introduces artificial random errors onto the static parameter $\theta$, thus converting and combining it with the other evolving parameters which are the states $x_{i}$ (see Liu and West, 2001)\nocite{Liu01}. Furthermore, in the ELW filter, the marginal posterior of $\theta$ at each time point is approximated by a finite mixture of normal distributions. The mean and variance of the evolution distribution are determined so that the mixture of normals does not increase the posterior variance. 
For the posterior update of the $\gamma$-set  of parameters, the idea of Storvik (2002)\nocite{storvik2002particle} and Fearnhead (2002)\nocite{fearnhead2002markov} is used. For the idea of ELW to be successfully applied, the posterior of $\gamma$, $p(\gamma \mid y_{1:i}, x_{0:i},\theta),\,\, i=1,\ldots,n$, needs to be tractable, that is from which samples can be drawn directly. 
In particular, we assume $p(\gamma \mid y_{1:i}, x_{0:i},\theta)$ depends on a sufficient statistic $s_i = s_i(y_{1:i}, x_{0:i},\theta)$.

Incorporating the evolution of $\theta$ into  (\ref{NDmodel}) according to the ELW methodology creates a further relaxation of the former model. The ELW model for the approximate dynamical model in (\ref{NDmodel}) is given by 
\begin{eqnarray}
\label{ELW1} y_i &\sim& N(x_i, \sigma^2 I_p), \\
\label{ELW2} x_i &\sim& N(g(x_{i-1}, t_i; \theta_i), u^2 I_p), ~~\mbox{and}\\
\label{ELW3} \theta_i &\sim& N(a\theta_{i-1} + (1-a)\bar{\theta}_{i-1},\tilde{h}^{2}V_{i}),
\end{eqnarray}
for $i=1,2,\cdots,n$ with $\theta_{0}\sim \pi_{\theta}$ and $x_{0}$ distributed according to its prior specification in (\ref{priors}).  In (\ref{ELW2}), $g$ is as defined in \eqref{RK4}, and $u$ is a small fixed positive real number representing the relaxation parameter. In (\ref{ELW3}),  $\bar{\theta}_{i-1}$  represents the posterior mean of $\theta$ given $y_{1:i-1}$ at time $i-1$, $a = (1-\tilde{h}^{2})^{1/2}$ where $\tilde{h}^{2} = 1 - ((3\delta - 1)/(2\delta))^2$, $\delta$ is a discounting factor usually taken to be a high value such as $0.95$ or $0.99$, and $V_{i}$ is the covariance matrix corresponding to the evolution equation of $\theta_{i}$. Equation (\ref{ELW3}) is the further relaxation and evolution model for $\theta$ prescribed by the ELW methodology (see Liu and West, 2001). The selection of the parameters $a$ and $\tilde{h}$ guarantees that the posterior variance of $\theta_i$ remains stable (i.e., does not increase) with the progression of the time index $i$.

Several posterior distributions will be needed for the subsequent discussion and we derive their forms here. Consider $\gamma = \lambda = \sigma^{-2}$, the inverse of the variance of  observation error. ELW methodology requires the distribution $p(\gamma\,|\,y_{1:i},x_{0:i},\theta)$ be tractable and easily sampled from. In our  case, the posterior distribution for $\gamma$, conditional on observations $y_{1:i}$, states $x_{0:i}$ and $\theta$, is given by
\begin{equation}
\label{gammaposterior}
\pi(\gamma \mid y_{1:i}, x_{0:i}, \theta)
= Gamma \left( a_\lambda + \frac{(i+1)p}{2}, \,\, b_\lambda + \frac{1}{2}\left(  \frac{\| x_0 - \mu_{x_0}\|^2}{c} + \sum_{k=1}^i \|y_k-x_k\|^2 \right) \right)
\end{equation}
which is a tractable distribution. Note also from the above equation that the distribution of $\gamma$ depends on $y_{1:i}$ and $x_{0:i}$ through the sufficient statistic $s_i= s_i(y_{1:i}, x_{0:i},\theta) = (a_\lambda + (i+1)p/2, b_\lambda + (\| x_0 - \mu_{x_0}\|^2/c + \sum_{k=1}^i \|y_k-x_k\|^2)/2 )$, where $a_{\lambda}, b_{\lambda}, c$ and $\mu_{x_0}$ are all fixed and known hyperparameters (see (\ref{priors})). Next, the two distributions, that is (i) the conditional distribution of $x_{i}$ given $x_{i-1}$, $y_{i}$, $\theta_{i}$ and $\gamma$, and (ii) the marginal distribution of $y_{i}$ given $x_{i-1}$, $\theta_{i}$ and $\gamma$, can be obtained by considering the joint density of $x_{i}$ and $y_{i}$, conditional on $x_{i-1}$, $\theta_{i}$ and $\gamma$, from (\ref{ELW1}) and (\ref{ELW2}). From these two equations, it follows that $(x_{i},y_{i})$ is jointly normal, and thus, the conditional density of $x_{i}$ given $y_{i}$ is 
\begin{equation}
\label{xigivenyi}
p(x_{i}\,|\,x_{i-1},y_{i},\theta_{i},\gamma) = N\left(\,\frac{y_{i}/\sigma^{2} + g(x_{i-1},t_{i},\theta_i)/u^2}{1/\sigma^{2} + 1/u^2}\,,\,\frac{1}{1/\sigma^{2} + 1/u^2}I_{p}\, \right),
\end{equation}
whereas the marginal distribution of $y_{i}$ given $x_{i-1}, \theta_{i}$ and $\gamma$, obtained by integrating out $x_{i}$, is given by
\begin{equation}
\label{marginalyi}
p(y_{i}\,|\,x_{i-1},\theta_{i},\gamma) = N\left(\, g(x_{i-1},t_{i},\theta_i),\,(\sigma^{2} + u^2)\,I_{p}\, \right).
\end{equation}


We now give the ELW algorithm for obtaining inference for $\theta$ based on the approximate dynamic model (\ref{NDmodel}) and the posteriors defined above. Let the notation $[A,B,\cdots\,|\,C,D,\cdots]$ denote the conditional density of random entities (either scalars or vectors) $A, B, \cdots$ conditional on either random or fixed constant entities $C, D,\cdots$. The ELW model of (\ref{ELW1})-(\ref{ELW3}) can be written based on this notation as
\begin{eqnarray}
\label{ELW1a} y_{i+1} &\sim& [\,y_{i+1}\,|\,x_{i+1}, \gamma\,], \\
\label{ELW2a} x_{i+1} &\sim& [\,x_{i+1}\,|\,x_{i},\theta_{i+1}\,], ~~\mbox{and}\\
\label{ELW3a} \theta_{i+1} &\sim& [\,\theta_{i+1}\,|\,\theta_{i},\,y_{1:i}\,].
\end{eqnarray}

Equation (\ref{ELW1a})-(\ref{ELW3a}) gives the joint distribution of $(y_{i+1},x_{i+1},\theta_{i+1})$ conditional on the observations, states and $\theta$-values at previous time points, that is,
\begin{displaymath}
[\,y_{i+1},\,x_{i+1},\,\theta_{i+1}\,|\,x_{i},\,\theta_{i},\,y_{1:{i}},\gamma\,] = [\,y_{i+1}\,|\,x_{i+1}, \gamma\,]\cdot [\,x_{i+1}\,|\,x_{i},\theta_{i+1}\,]\cdot [\,\theta_{i+1}\,|\,\theta_{i},\,y_{1:i}\,]
\end{displaymath}
based on (\ref{ELW1a})-(\ref{ELW3a}). The auxiliary particle filter (APF) technique rewrites this joint density as
\begin{equation}
\label{APF}
[\,y_{i+1},\,x_{i+1},\,\theta_{i+1}\,|\,x_{i},\,\theta_{i},\,y_{1:{i}},\gamma\,] \\ 
=
{[\,x_{i+1}\,|\,x_{i},\theta_{i+1},\,y_{i+1},\,\gamma\,]}\cdot[\,y_{i+1}\,|\,x_{i},\theta_{i+1},\gamma\,]\cdot[\,\theta_{i+1}\,|\,\theta_{i},y_{1:{i}}\,].
\end{equation}
The first term on the right hand side of (\ref{APF}) is given by (\ref{xigivenyi}), thus available in closed form for sampling in our examples.  The second term on the right hand side of  (\ref{APF}) is given by (\ref{marginalyi}), which is again available in closed form for {\it evaluation} in our examples. The third term in (\ref{APF}) is the Liu and West filter for $\theta$ given by (\ref{ELW3a}), which can be easily sampled from. We give our sampling methodology to sample from the posteriors using sequential Monte Carlo. Suppose $\{x_{i}^{(j)},\,\theta_{i}^{(j)},\, \gamma_{i}^{(j)},\,s_{i}^{(j)}\}$ for $j=1,2,\cdots,N$ are $N$ samples from the posterior  $\,[x_{i},\,\theta_{i},\, \gamma_{i},\,s_{i}\,|\,y_{1:i}\,]$. The subscript $i$ on $\gamma_{i}$ does not imply any evolution equation for $\gamma$. It just denotes the random variable $\gamma$ for marginal realizations of $\gamma$ from  the posterior $[\gamma\,|\, s_i]$. Similarly, $s_{i}$ denotes realizations of the sufficient statistic at time point $i$ based on its functional equation, namely, $\mathcal{S}(y_{1:i},x_{0:i},\theta_{i})$ when $x_{0:i}$ and $\theta_{i}$ are samples from the posterior $[x_{0:i},\theta_{i}\,|\,y_{1:i}]$.

The steps of our sampling algorithm is as follows:
\begin{itemize}
	\item First, sample $\theta_{i+1}^{(j)} \sim [\theta_{i+1}\,|\,\theta_{i}^{(j)},\,y_{1:i}] $ according to (\ref{ELW3}) for $j=1,2,\cdots,N$.
	\item Compute weights $w_{i}^{(j)} \propto [\,y_{i+1}\,|\,x_{i}^{(j)},\,\theta_{i+1}^{(j)},\,\gamma_{i}^{(j)}\,]$ for $j=1,2,\cdots,N$.
	\item Obtain $N$ resamples $\{\,\tilde{x}_{i}^{(j)},\,\tilde{\theta}_{i+1}^{(j)},\,\tilde{\gamma}_{i}^{(j)},\,\tilde{s}_{i}^{(j)}\,\}_{j=1}^{N}$ by sampling from the collection $\{\,{x}_{i}^{(j)},\,{\theta}_{i+1}^{(j)},\,{\gamma}_{i}^{(j)},\,{s}_{i}^{(j)}\,\}_{j=1}^{N}$ according to the weights $\{\,w_{i}^{(j)}\,\}_{j=1}^{N}$. 
	\item Sample $\tilde{x}_{i+1}^{(j)} \sim [\,x_{i+1}\,|\,\tilde{x}_{i}^{(j)},\,\tilde{\theta}_{i+1}^{(j)},\,y_{i+1},\,\tilde{\gamma}_{i}^{(j)}\,]$ for $j=1,2,\cdots,N$.
	\item Compute $\tilde{s}_{i+1}^{(j)} = \mathcal{S}(\tilde{s}_{i}^{(j)},\,y_{i+1},\,\tilde{x}_{i+1}^{(j)},\,\tilde{\theta}_{i+1}^{(j)})$ for $j=1,2,\cdots,N$.
	\item Sample $\tilde{\gamma}_{i+1}^{(j)} \sim [\,\gamma\,|\,\tilde{s}_{i+1}^{(j)}]$ for $j=1,2,\cdots,N$.
\end{itemize}
Then, it follows that the $N$ samples $\{\tilde{x}_{i+1}^{(j)},\,\tilde{\theta}_{i+1}^{(j)},\, \tilde{\gamma}_{i+1}^{(j)},\,\tilde{s}_{i+1}^{(j)}\}$ for $j=1,2,\cdots,N$ are realizations from the posterior $\,[x_{i+1},\,\theta_{i+1},\, \gamma_{i+1},\,s_{i+1}\,|\,y_{1:i+1}\,]$. As the tuning parameter $\tilde{h} \rightarrow 0$, the posterior of $\theta$ at every time point $i$ from the approximate dynamic model becomes closer to the true posterior from the DEM.

As mentioned earlier, in the above algorithm, the subscripts $i$  on $\gamma_i$ and $s_i$ do not imply any kind of evolution over time. They just represent the update of the parameter and statistic, respectively, as new data become available. The tuning parameter $a$ determines the extent of shrinkage of the normal mixture through its mean. It also controls the smoothness through the variance term $\tilde{h}^2V_{i}$. It is usually prescribed to be chosen around the value $0.95$. The tuning parameter $a$ was fixed at $0.95$ throughout the rest of examples.  This corresponds to taking $\tilde{h}^{2} = 1-a^2 = 0.0975$ and $\delta = 1/(3-2a) = 0.909$. For the covariance matrix $V_i$, we chose $V_i = (N-1)^{-1} \sum_{j=1}^N (\theta_{i-1}^{(j)} - \bar{\theta}_{i-1}) (\theta_{i-1}^{(j)} - \bar{\theta}_{i-1})^T$.

The initial proposal density $q(x_0, \theta, \gamma)$ affects the performance of the algorithm. The proposal density which is concentrated around the true parameter has a better performance than the other proposal densities even with relatively small number of particles. In practice, we suggest that one run the ELW filter with initial particles $\theta^{(j)}$ and $\gamma^{(j)}$ from $\pi(\theta, \gamma)$ and rerun with the particles $\hat{\theta}^{(j)}$ and $\hat{\gamma}^{(j)}$ from the first inference. It is equivalent to consider the proposal density
$$q(x_0, \theta, \gamma) \equiv \pi(x_0) \times \pi(\theta ,\gamma \mid {\bf y}_n).$$
We call the resulting particles the refined particles.
It was used throughout the rest of examples.


\section{Convergence of the Posterior}\label{section:Conv}

\subsection{Convergence of the Posterior as the relaxiation parameter decreases}

In this subsection, we show that as the relaxation parameter $u$ converges to $0$, the  posterior density of  $(x_0, \theta, \lambda)$ from the approximate dynamic model converges to the posterior from the DEM, i.e.
\bean\label{postofDM}
\pi(x_0, \theta, \lambda | {\bf y}_n , u^2) = \frac{\int L(\Lambda) \pi (dx_1, \ldots, dx_n | x_0, \theta, u^2) \pi(x_0, \theta, \lambda)}{\int \int L(\Lambda) \pi (dx_1, \ldots, dx_n | x_0, \theta, u^2) \pi(dx_0, d\theta, d\lambda)}
\eean
converges to
\bean\label{postofODE}
\pi (x_0, \theta, \lambda | {\bf y}_n) = \frac{L^*(x_0, \theta, \lambda)\pi(x_0, \theta, \lambda) }{\int L^*(x_0, \theta, \lambda) \pi(dx_0, d\theta, d\lambda)} 
\eean
as $u^2 \to 0$, where $\Lambda = (x_0, \ldots,  x_n, \theta, \lambda)$,
\bea
L(\Lambda) &=& (\lambda)^{{np}/{2}} \exp \left({-\frac{\lambda}{2}\cdot \sum_{i=1}^n \|y_i - x_i \|^2} \right) \text{ and } \\
L^*(x_0, \theta, \lambda) &=& (\lambda)^{{np}/{2}} \exp \left({-\frac{\lambda}{2}\cdot \sum_{i=1}^n \|y_i - g^{i}(x_0, t_{i-1} ; \theta) \|^2} \right)
\eea
with $g^{i}(x_0, t_{i-1} ;\theta) = g( g^{i-1}(x_0, t_{i-2};\theta), t_{i-1} ; \theta)$. Note that $\pi (x_0, \theta, \lambda | {\bf y}_n)$ is the  posterior of DEM.

\begin{theorem}\label{convtoDE}
Consider model \eqref{NDmodel} and prior \eqref{priors}. Suppose $f( x, t ; \theta)$ is continuous in $x$. Then, the posterior density of the dynamic model \eqref{NDmodel} converges to that of the differential equation model \eqref{ODEmodel}, i.e.
$$\pi(x_0, \theta, \lambda | {\bf y}_n , u^2 ) \to \pi(x_0, \theta,\lambda | {\bf y}_n)$$
for all $x_0, \theta, \lambda$ as $u^2 \to 0$.
\end{theorem}

%
\subsection{Convergence of the Posterior as the step size increases}
We have shown that the posterior of the dynamic model \eqref{NDmodel} converges to that of the differential equation model \eqref{ODEmodel} as $u^2 \to 0$. In this subsection, we will prove that the posterior of the differential equation model converges to that of the true model. 

If the step size is $m$, each time interval $[t_{i-1},t_i]$ is divided into $m$ segments of length $(t_i - t_{i-1})/m$, and the Runge-Kutta method is applied to each subinterval to obtain $x_i's$.  To clarify the difference, let $x^m$ be the approximated solution of the differential equation by the fourth-order Runge-Kutta method with $m$ segments. Similarly, let $\pi_m$ and $\pi_{true}$  be the posterior distributions corresponding to $x^m$ and the true $x$, respectively. 
Note $x^m(t_1) = x(t_1)$ for all $m$. 

\begin{theorem}\label{convtoTrue}
Consider model \eqref{ODEmodel} and prior \eqref{priors}. Suppose $f( x, t ; \theta)$ satisfies Lipschitz condition in $x$, i.e. there exists the constant $K > 0$ such that
\begin{equation}\label{Lipschitz}
\| f(x, t ; \theta) -  f(x', t ; \theta) \| < K \| x - x' \| 
\end{equation}
for any $x, x' \in \bbR^p, t \in [T_0, T_1]$ and $\theta \in \Theta$. Then, the posterior density of the differential equation model \eqref{ODEmodel} converges that of the true model, i.e.
$$\pi_m (x_0,\theta,\lambda | {\bf y}_n) \to \pi_{\text{true}} (x_0,\theta,\lambda | {\bf y}_n)$$
for all $x_0, \theta, \lambda$ as $m \to \infty$.
\end{theorem}

This result guarantees that the differential equation model works well with a reasonable segments parameter $m$ under the Lipschitz condition. 

\subsection{Choice of the relaxation parameter and the step size}\label{subsection:u2andm}
In practice, the choice of $u^2$ and $m$ can affect the performance of the approximation. The approximate posterior distribution may vary  by different choice of these values.
Theoretically, the smaller the relaxation parameter $u^2$ is, the closer the approximate posterior is to the true posterior. But in practice we may need moderately large value of $u^2$ to get stable posterior approximation. 
We suggest following strategy for choosing the variance of state $u^2$.
Consider various $u^2$ values from large to small values in turn. For each $u^2$ value, check the stability of posteriors by running two or three ELW filters simultaneously. Here, the stability means that all posterior densities based on ELW runs are closed enough to each other. Finally, use the smallest $u^2$ value  for the inference which gives the stable result.


For convenience, let $h \equiv t_{i+1}-t_i$ for all $i=1,2,\ldots, n-1$. For the choice of $m$, we assume  $h/m = O(n^{-\alpha})$. Theoretically, the larger value of $m$ gives more  accurate inference, but it would require heavier computation.
In the following theorem, we relate the step size $h/m$ to the approximation error rate of the posterior, and  based on the theorem we suggest values of $m$ for computation according to the acceptable error rate. The theorem requires the following assumptions.
\begin{itemize}
	\item[A1.] $\{x(t): t\in [0,T]\}$ is a compact subset of $\mathbb{R}^p$;
	\item[A2.] $\{y(t): t\in [0,T]\}$ is a bounded subset of $\mathbb{R}^p$; and
	\item[A3.] the $K$th order derivative of $f(x,t;\theta)$ with respect to $t$ exists and is continuous in $x$ and $t$, where $K$ is the order of the numerical method $g$.
\end{itemize}
\begin{theorem}\label{errorrate}
Consider model \eqref{ODEmodel} and prior \eqref{priors}. Suppose $f( x, t ; \theta)$ satisfies Lipschitz condition \eqref{Lipschitz} in $x$, and suppose $A1-A3$ hold. Let $K$ be the order of the numerical method $g$ and $h/m = O(n^{-\alpha})$. 
If $\alpha \ge (1+R)/K$, 
the error rate of the posterior approximation is $O(n^{-R})$ for sufficiently large $n$, i.e.,
$$\pi_m(x_0, \theta, \lambda | {\bf y_n}) = \pi(x_0,\theta,\lambda | {\bf y_n}) \times (1+O(n^{-R})) $$
for all $x_0,\theta,\lambda$, then $\alpha \ge (1+R)/K$ is sufficient.
\end{theorem}
 Note that the order of Runge-Kutta method is 4, and the rate of $h$ is $n^{-1}$ because we consider a bounded time interval $[0,T]\subset \bbR$ with $T<\infty$. 
By the above theorem, if we want to get the error rate $O(n^{-3})$ or larger, we know that it can be achieved by $m=1$ for large $n$. However, in practice, one should notice that the additional error from the SMC sampling may arise. In such case, we may need to use $m$ bigger than $1$.


\section{Simulated Data Examples}\label{section:examples}
\subsection{Newton's law of Cooling}
\subsubsection{Description of model and data generation step}
Newton's law of cooling, made by English physicist Isaac Newton, is a model describing the temperature change of an object. 
According to the model, the temperature of an object changes  proportional to the temperature difference between the object and its surroundings. This notion is given by the following ODE form 
\begin{eqnarray}\label{NLC}
{\dot x}(t) & = & \theta_1(x(t) - \theta_2),
\end{eqnarray}
where $x(t)$ is the temperature of the object at time $t$, $\theta_1$ is a negative constant and $\theta_2$ is the temperature of the surroundings. All of the temperature are in Celcius. For more details, see Incropera (2006).\nocite{Incropera06}

We chose this model  as a testbed for our method. Since the solution of \eqref{NLC} is known as
\bean\label{NLCtrue}
x(t) = \theta_2 - (\theta_2 - x_0)e^{\theta_1 t}
\eean
where $x_0 = x(0)$,  we can calculate the true posterior  directly. 
The data $y_i = y(t_i)$ was generated with the true mean function \eqref{NLCtrue} and we set the model parameters as $x_0 = 20$, $\theta = (-0.5, 80)^T$, $\sigma^2 = 25$ and time points $t_i = i h $ for $i= 1,\ldots, n$ where the sample size $n =100$ and the step size $h = 0.15$. The simulated data and the true mean function are shown in Figure \ref{fig:NLCm}.
\begin{figure}[!tb]
		\centering
		\includegraphics[width=13cm, height=11cm]{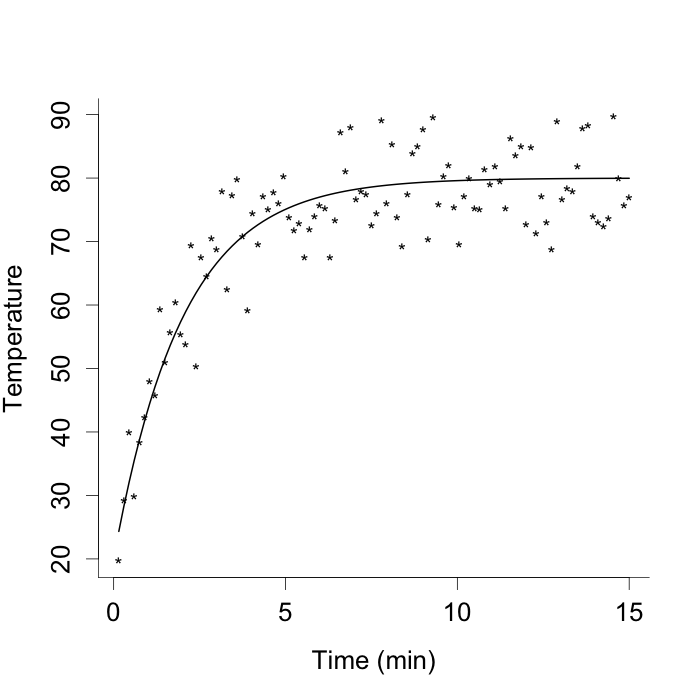}
		\caption{The solid line is the true temperature as a function of time from the Newton's law of cooling model with $x_0 =$ 20, $\theta = (-$0.5, 80$)^T$. The star-shaped points are the generated data of temperatures with $\sigma^2=25$.}
		\label{fig:NLCm}
\end{figure}

The priors were set by
\begin{eqnarray*}
x_0 \mid \lambda & \sim& N(\mu_{x_0} , c/\lambda) \\
\lambda & \sim& Gamma(a_\lambda, b_\lambda) \\
\theta =(\theta_1, \theta_2)  & \sim& Uniform \left((-100, 0)\times (50, 150) \right)
\end{eqnarray*}
where $\mu_{x_0} = y_1, a_\lambda=1, b_\lambda=1$ and $c = 1$. The values of $y_i$ are in  the interval $[65, 90]$ after 50th observation, and the temperature of the surroundings, $\theta_2$, must be the around the interval. The prior of $\theta_2$ is set by $Uniform (50, 150)$ whose support includes $[65, 90]$. With a similar reasoning, we set $\theta_1 \sim Uniform(-100, 0)$. 


The true posterior of $\theta$ and $\lambda$ can be obtained as follows:
\begin{eqnarray}\label{truePost}
\begin{split}
\lambda \mid \theta, y_{1:n} &\sim Gamma(\frac{np}{2}+a_\lambda, \frac{1}{2}\tilde{u}(\theta) + b_\lambda)\\
\theta \mid y_{1:n} &\sim \frac{1}{(\frac{1}{2}\tilde{u}(\theta)+b_\lambda)^{\frac{np}{2}+a_\lambda} } I(-100 < \theta_1 < 0) I(50<\theta_2 < 150),
\end{split}
\end{eqnarray}
where
\begin{eqnarray*}
\tilde{u}(\theta) &=& \mu_{x_0}^2/c + \sum_{i=1}^n z_i^2 - (1/c + \sum_{i=1}^n e^{2\theta_1ih})^{-1}(\mu_{x_0}/c + \sum_{i=1}^n z_ie^{\theta_1ih})^2, \\
z_i &=& z_i(\theta) =  y_i - \theta_2 + \theta_2 e^{\theta_1ih}.
\end{eqnarray*}

\subsubsection{Assessment of the convergence of the posteriors}
We assessed the convergence of posteriors which is described at Theorem \ref{convtoDE}.
To show that the posterior of dynamic model converges to that of DEM, we got the simulation results for RDEM with $u^2= 1, 0.1^1, 0.1^2$ and $0.1^5$. 
The DEM was treated as a dynamic model with small value of $u^2$. 
We ran the ELW filter based on 20,000 particles and fixed the number of segments $m$ at 1. For all of the settings, the ELW filter takes less than 3 seconds for 20,000 particles. 
The histogram of the marginal posterior distributions are drawn at Figure \ref{fig:NLCu2}. 
It seems that the posterior of dynamic model approaches that of the DEM as $u^2$ decreases to zero. Thus, it supports the theoretical result, Theorem \ref{convtoDE}. 
\begin{figure}[!tb]
	\centering
	\includegraphics[width=15cm, height=7cm]{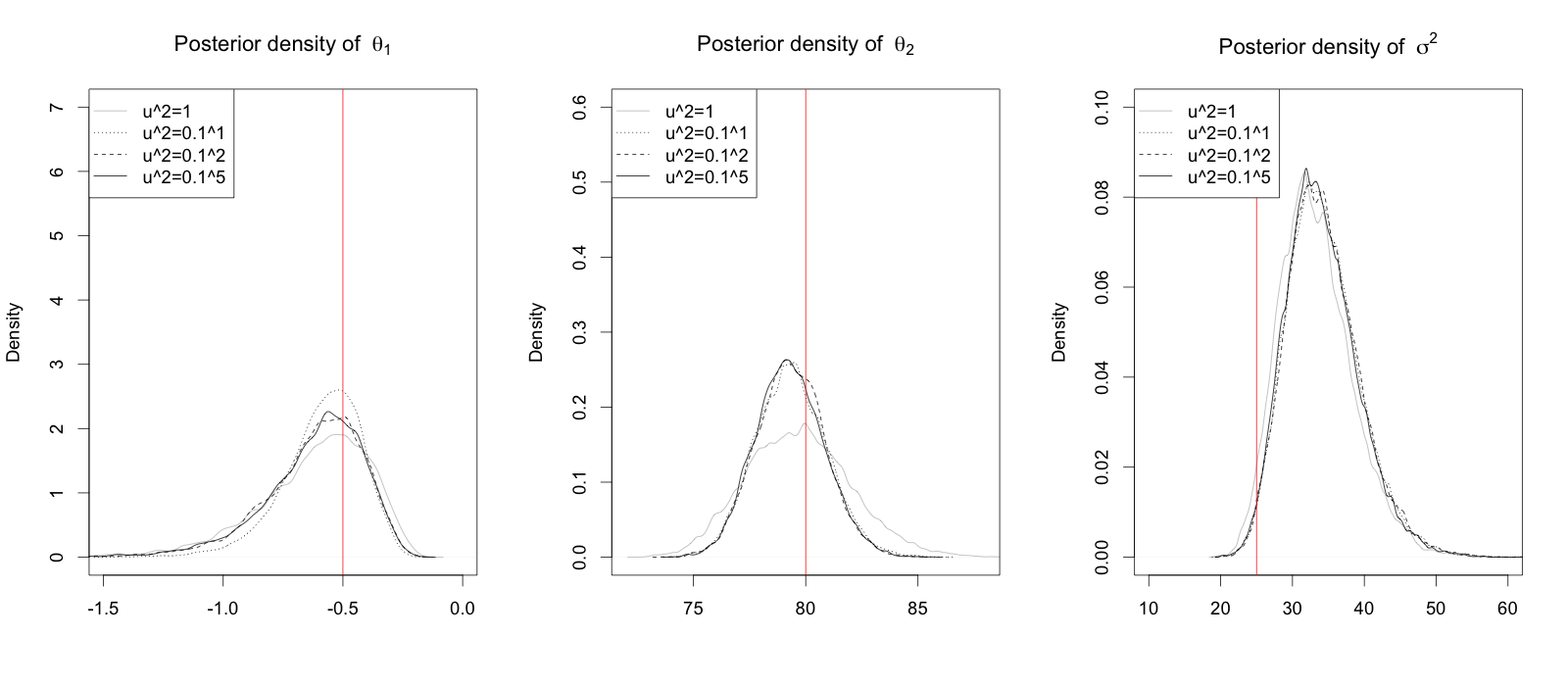}
	\vspace{-.5cm}
	\caption{The histograms of the marginal posterior distributions of the dynamic models with $u^2=1, 0.1, 0.1^2, 0.1^5$ and $m=1$ from the Newton's law of cooling. The red lines are the true values of parameters, $(\theta_1, \theta_2, \sigma^2) = (-0.5, 80, 25)$.}
	\label{fig:NLCu2}
\end{figure}

To show that the posterior of DEM converges to that of true model, we got the simulation results for the DEM with the number of segments $m=1,2,4$ and the true model. 
We approximated DEM by the dynamic model with $u^2= 0.1^5$.
For the true model, we used a grid sampling algorithm for the true posterior \eqref{truePost}. For each setting, the ELW filter takes less than 3 seconds for 20,000 particles. 
The grid set was chosen by $[-2, 0]\times[70,90]$, and each axis was divided into 50 equal length intervals resulting 51 points. 20,000 posterior samples were drawn.
The histograms of the marginal posterior distributions are drawn at Figure \ref{fig:NLCminf}. 
The posterior densities of DEM are quite similar to each other, but they have the larger variation than the true posterior densities. 
\begin{figure}[!tb]
	\centering
	\includegraphics[width=15cm, height=7cm]{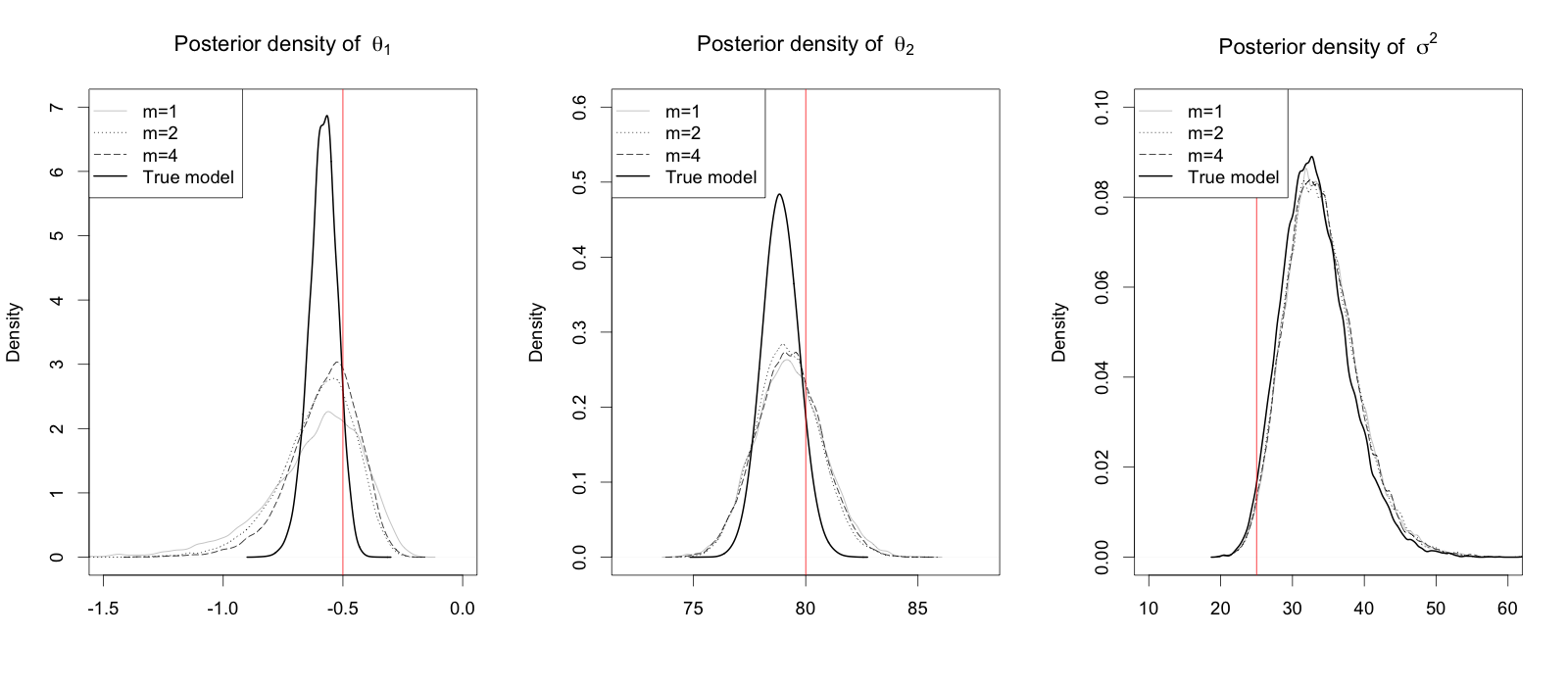}
	\vspace{-.5cm}
	\caption{The histograms of the marginal posterior distributions of the dynamic models with $u^2 =0.1^5$ and $m=1,2, 4$, and those of the true model from the Newton's law of cooling. 
	The red lines are the true values of parameters, $(\theta_1, \theta_2, \sigma^2) = (-0.5, 80, 25)$.}
	\label{fig:NLCminf}
\end{figure}

%

\subsection{FitzHugh-Nagumo model}\label{subsection:LV}
\subsubsection{Description of model and data generation step}\label{subsubsection:LVdesc}
FitzHugh-Nagumo model (FitzHugh, 1961\nocite{Fitzhugh1961impulses}; Nagumo et al. 1962\nocite{Nagumo1962active}) describes the action of spike potential in the giant axon of squid neurons by an ODE with two state variables and three parameters:
\bea
\dot{x}_1(t) &=& \theta_3 \left( x_1(t) - \frac{1}{3}x_1^3(t) + x_2(t) \right) , \\
\dot{x}_2(t) &=& -\frac{1}{\theta_3} \left( x_1(t) - \theta_1 + \theta_2 x_2(t)  \right),
\eea
where $-0.8 < \theta_1, \theta_2 < 0.8$ and $0 < \theta_3 <8$. The two state variables, $x_1(t)$ and $x_2(t)$, are the voltage across an membrane and outward currents at time $t$, respectively. 

\begin{figure}[!tb]
	\centering
	\includegraphics[width=0.65\textwidth]{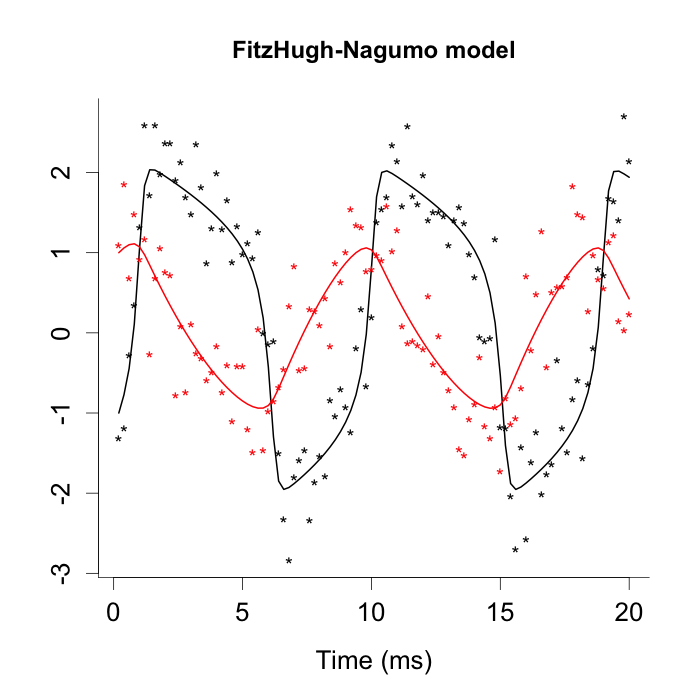}
	\caption{The solid lines are $x_1$ (black line) and $x_2$ (red line) as a function of time from the FitzHugh-Nagumo model with $x(t_0) = (-1, 1)^T, \theta = (0.2, 0.2, 3)^T$. The star-shaped points are the generated data of the populations with $\sigma^2=0.25$.}
	\label{fig:FNmod}
\end{figure}
Using the FitzHugh-Nagumo model, we compare the proposed method with the parameter cascading method (Ramsay et al., 2007), the delayed rejection adaptive Metropolis (DRAM) algorithm (Soetaert and Petzoldt, 2010)\nocite{soetaert10} and the Laplace approximated posterior (LAP) method (Dass et al., 2017)\nocite{dass2017laplace}.
The data $y_i = y(t_i)$ was generated from DEM \eqref{ODEmodel} with the model parameters $x_0 = (-1, 1)^T, \theta = (0.2, 0.2, 3)^T$, $\sigma^2 = 25$ and time points $t_i = i h $ for $i= 1,\ldots, n$, where the sample size $n = 100$ and the step size $h = 0.2$, $m=400$. 
The simulated data and the true mean function are shown in Figure \ref{fig:FNmod}.
The priors were set by 
\begin{eqnarray*}
x_0 \mid \lambda &\sim& N(\mu_{x_0} , c\lambda^{-1} I_2) \\
\lambda &\sim& Gamma(a_\lambda, b_\lambda)\\
\theta  &\sim& Uniform(A)
\end{eqnarray*}
where $\mu_{x_0} = y_1, a_\lambda=1, b_\lambda=1$, $c = 1$ and $A = \{(\theta_1,\theta_2,\theta_3): -0.8 <\theta_1,\theta_2<0.8 , 0< \theta_3< 8 \}$.

\subsubsection{Comparison with other methods}
To compare  the proposed method (RDEM-ELW) with other methods, the parameter cascading (PC) method, DRAM algorithm and LAP method were applied to the same data set. We used the \verb|R| packages \verb|CollocInfer| and \verb|FME| for the parameter cascading and DRAM, respectively. 

The PC method is one of the popular frequentist methods for estimating the parameters in ODE. It uses the collocation method which represents the state vector $x(t)$ as a series of basis expansion. The penalized likelihood criterion has three components: the matrix of coefficients of basis expansions $C$, the unknown parameter $\theta$ and the smoothing parameter $\lambda$. PC optimizes the penalized likelihood by two steps. In the inner optimization, the criterion is optimized with respect to the coefficient $C$ while $\theta$ and $\lambda$ are fixed. After that, in the outer optimization, the penalized likelihood is optimized with respect to $\theta$ while $\lambda$ is kept fixed. The smoothing parameter $\lambda$ is chosen based on the appropriate criteria such as the numerical stability of parameter estimates or the forward prediction error (Hooker et al., 2000)\nocite{hooker10}. For more details about PC method, see Ramsay et al. (2007)\nocite{RamsayHooker07}.
For the PC method, we used the third-order B-spline basis and $2n-1$ equally spaced knots on $[t_0, t_n]$. The smoothing parameter was set by $\lambda = 10^5$. The initial parameter were drawn from $N(\theta_0, (0.01)^2I_q)$ where $\theta_0$ is the true parameter value. 

The DRAM algorithm, a variant of the standard Metropolis-Hastings algorithm (Metropolis et al., 1953;\nocite{metropolis53} Hastings, 1970\nocite{hastings70}), is chosen as a benchmark in the Bayesian side. 
With the \verb|R| package \verb|FME| (Soetaert and Petzoldt, 2010)\nocite{soetaert10package}, one can infer the DEM with DRAM algorithm for the parameters and numerical integration for the state variables. 
We applied the DRAM algorithm with the initial parameter as the maximum likelihood estimate using \verb|modFit()| function and the maximal number of tries 1. The parameter covariance was updated in every 100 iteration. 
We got 20,000 posterior samples for the inference. 

LAP method is another benchmark in the Bayesian side. It is fast when the dimension of parameter is small and empirically has comparable or better performance than PC method and DRAM algorithm (Dass et al., 2017)\nocite{dass2017laplace}. Since the dimension of parameter is small, the grid sampling method for $\theta$ was chosen. For each parameter $\theta_i$, the grid range was chosen by $[\widehat{\theta}_i^R \pm 4 \widehat{sd}(\widehat{\theta}_i^R)]$ where $\widehat{\theta}_i^R$ is the parameter estimate for $\theta_i$ from the PC method. Each axis was divided into $31$ intervals of equal length, and the step size for numerical integration was set at $m=2$. The priors for parameters were set as in subsection \ref{subsubsection:LVdesc}, and 20,000 posterior samples were obtained.

For the RDEM-ELW, the step size for numerical integration and the variance for the state were chosen by $m=2$ and $u^2= 0 .1^5$, respectively. 
The priors for parameters were set as described in subsection \ref{subsubsection:LVdesc}, the number of particles was chosen by $N=20,000$.
We generated 100 simulated data set using the 4th order Runge-Kutta. The model parameters were set as described in subsection \ref{subsubsection:LVdesc}.

For RDEM, PC and DRAM methods, \verb|R| and C/C$++$ were used for implementation. \verb|R| and Fortran90 were used for LAP method. On average based on 100 simulations, it took only 3.523 seconds for estimation, while the PC method, DRAM algorithm and LAP method took 49.152, 276.700 and 215.591 seconds, respectively. The boxplot of computation times for each method is given at Figure \ref{fig:compare_times}. The proposed RDEM method significantly reduced the computation time. It was even faster than the frequentist method, the PC method. 
Thus, the RDEM method has an enormous advantage in computation speed over other methods. 
Table \ref{table:compare_rmse} represents the absolute biases, standard deviations for $\hat{\theta}$ and root mean squared errors (rmse) for $\hat{\theta}$ in the FitzHugh-Nagumo model. It seems RDEM method provides reasonable estimates in terms of bias, but larger standard deviation than others.

\begin{figure}[!t]
	\centering
	\includegraphics[width=10cm, height=9cm]{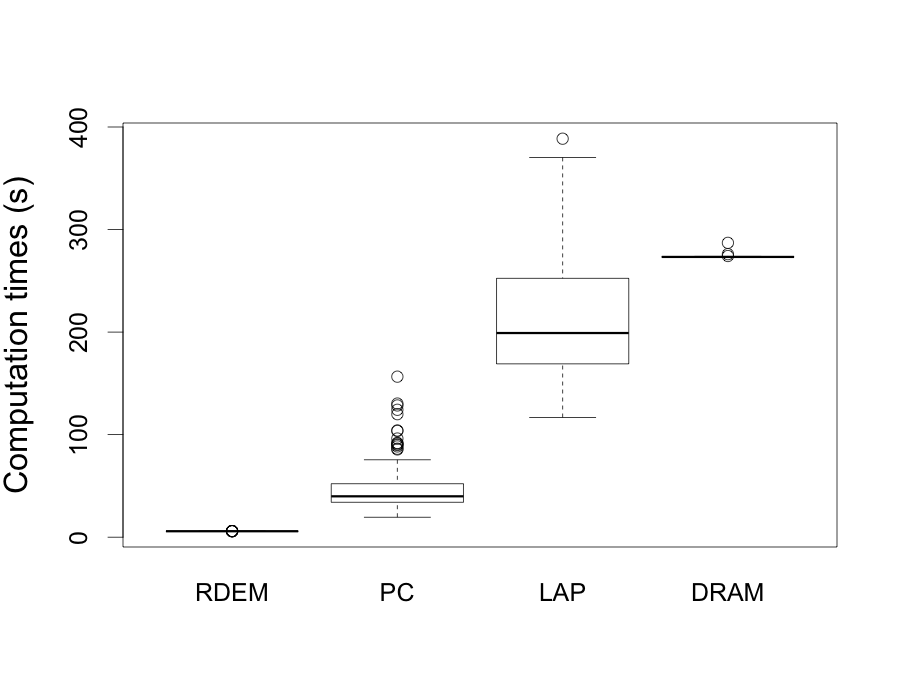}\vspace{-.5cm}
	\caption{The boxplots of the computation times for $\hat{\theta}$ based on 100 simulated date sets. The results for the relaxed DEM with ELW filter (RDEM), the parameter cascading (PC)method, Laplace approximated procedure (LAP) and delayed rejection adaptive Metropolis (DRAM) algorithm are shown.}
	\label{fig:compare_times}
\end{figure}
\begin{table}[!tb]
	\scriptsize\centering\vspace{0cm}
	\caption{The table of mean of the absolute biases, standard deviations and root mean squared errors (rmse) for $\hat{\theta}$ in the FitzHugh-Nagumo model. The results for the relaxed DEM with ELW filter (RDEM), parameter cascading (PC)method, Laplace approximated procedure (LAP) and delayed rejection adaptive Metropolis (DRAM) algorithm are shown.}
	\begin{tabular}{|c|c|c|c|c|c|}\hline
		\multicolumn{2}{|c|}{} & RDEM & PC & LAP & DRAM  \\ \hline 
		\multirow{3}{*}{Absolute bias} & $\theta_1$ & 0.051  & 0.024  & 0.024  & 0.024  \\ 
		& $\theta_2$ & 0.135 & 0.106 & 0.099  & 0.100  \\ 
		& $\theta_3$ &  0.108 & 0.039 & 0.044 & 0.047 \\ \hline
		\multirow{3}{*}{Standard deviation} & $\theta_1$ &  0.063  & 0.027  & 0.027  & 0.028   \\ 
		& $\theta_2$ & 0.130  & 0.123  & 0.117  & 0.119  \\ 
		& $\theta_3$ & 0.194 & 0.060 & 0.056 & 0.059 \\ \hline
		\multirow{3}{*}{rmse} & $\theta_1$ & 0.084  & 0.038  & 0.038  & 0.040  \\
		& $\theta_2$ & 0.198  & 0.171  & 0.161  & 0.164  \\ 
		& $\theta_3$ & 0.233 & 0.076 & 0.075 & 0.079  \\ \hline
	\end{tabular}\label{table:compare_rmse}
\end{table}
\section{Lynx-hare data: Lotka-Volterra equation}\label{sec:realdata}
There are large number of models to express predator-prey relationships because predation is often direct, conspicuous and easy to study. Lotka-Volterra model is one of the simplest model of predator-pray interactions. Lotka (1925)\nocite{Lotka25} and Volterra (1926)\nocite{Volterra27} independently developed the model of the form: 
\begin{eqnarray}\label{LVeq}
	\begin{split}
		\dot{x}_1(t) &=& x_1(t) (\theta_1 - \theta_2 x_2(t) ), \\
		\dot{x}_2(t) &=& -x_2(t) (\theta_3 - \theta_4 x_1(t) ), 
	\end{split}
\end{eqnarray}
where $x_1$ denotes the number of preys, and $x_2$ denotes the number of their predators. The model parameters $\theta_1, \theta_2, \theta_3$ and $\theta_4$ are the intrinsic rate of prey population increase, the predation rate, the predator mortality rate and the offspring rate of the predator, respectively. 

\begin{figure}[!tb]\centering
\begin{tabular}{ccc}
		\includegraphics[width=12cm, height=10cm]{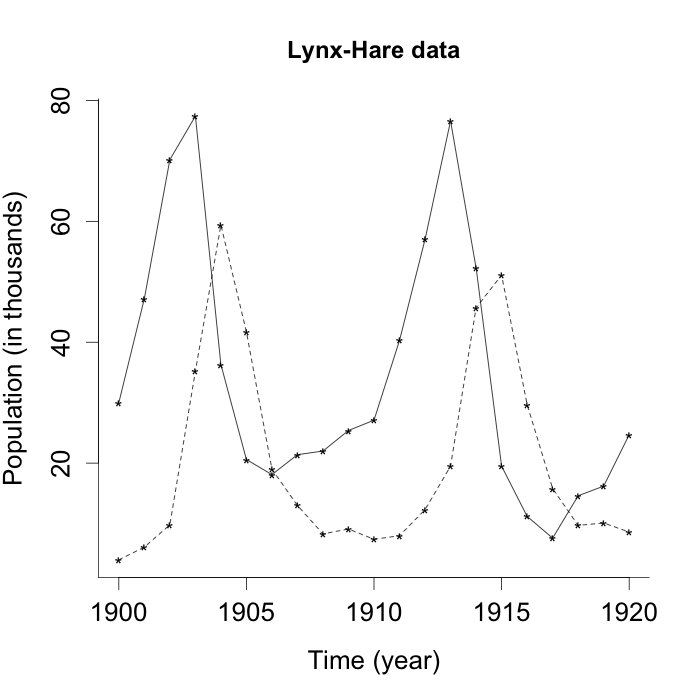}
\end{tabular}
\vspace{-0.5cm}
\caption{The numbers of trapped lynx and snowshoe hares between 1900 and 1920 is drawn. The solid line is the number of hares, and the dotted line is the number of lynx.}
\label{fig:LynxHare}
\end{figure}
Lynx-hare data is a popular data set representing the number of captured lynx and snowshoe hares in North Canada which was collected by Hudson Bay company. It contains the number of furs of lynx and hares, so it implies the actual populations of them. We obtained the annual data  between 1900 and 1920 recorded in thousands from Li (2012)\nocite{Li12} which is given at Figure \ref{fig:LynxHare}. The Lotka-Volterra equation, the equation \eqref{LVeq}, is fitted to the data set and used to predict the future values of trapped lynxes and hares. 

\begin{figure}[!tb]\centering
\begin{tabular}{ccc}
		\includegraphics[width=16cm, height=10cm]{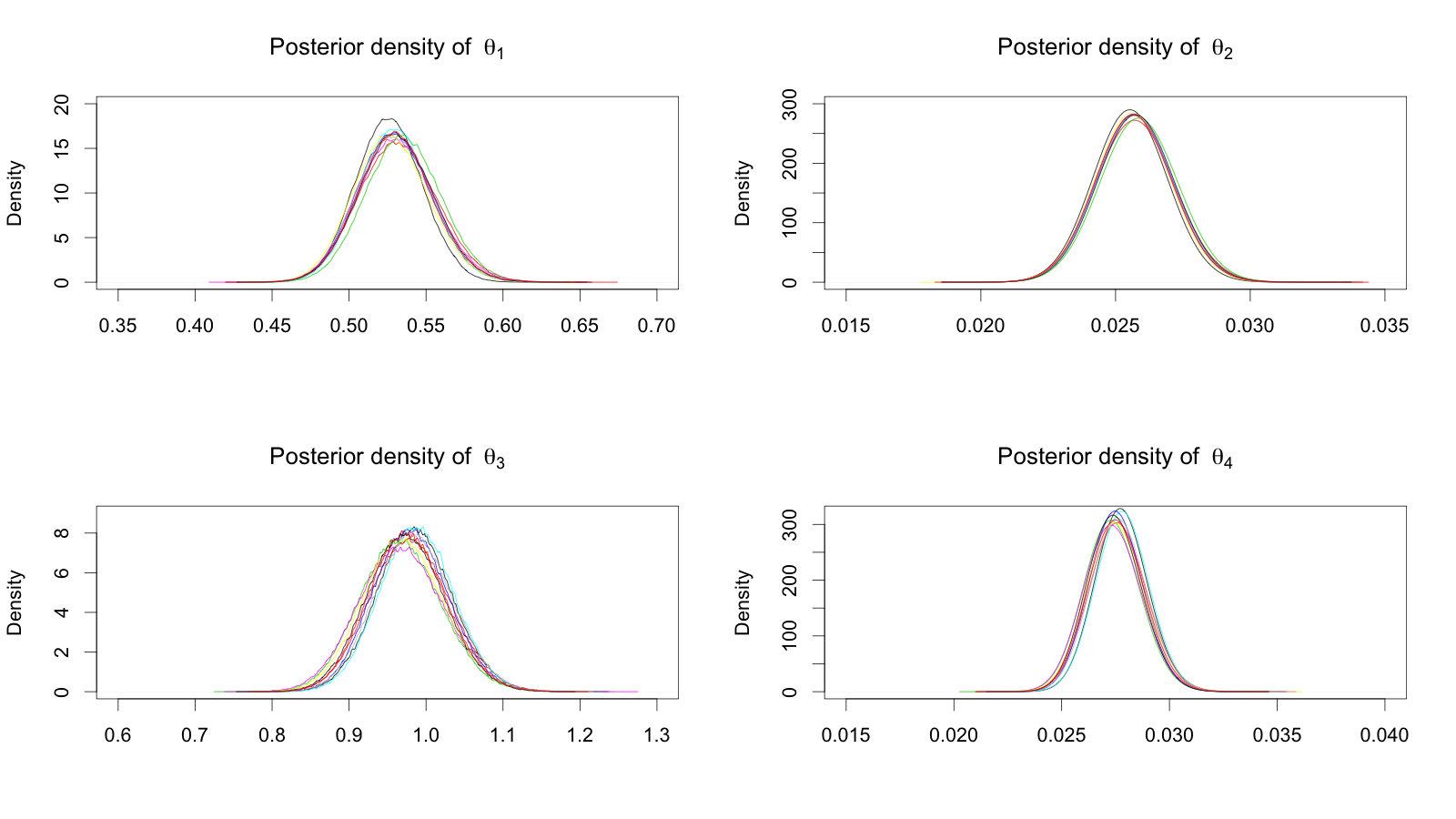}
\end{tabular}
\vspace{-1cm}
\caption{The posterior densities of the Lotka-Volterra equation for the lynx-hare data based on 10 ELW filter runs with $m=2$ and $u^2 = 5$.}
\label{fig:LV_m_real}
\end{figure}
\begin{table}[!tb]
\scriptsize\centering
\caption{Posterior summary statistics for the parameter of the Lotka-Volterra equation for the lynx-hare data with $m=2$ and $u^2=10$.}
\begin{tabular}{|c|c|c|c|}\hline
 & Mean & Median & 90\% credible interval  \\ \hline
$\theta_1$ & 0.526  & 0.525  & (0.491, 0.562) \\
$\theta_2$ & 0.026  & 0.026  & (0.024, 0.027) \\
$\theta_3$ & 0.986  & 0.985  & (0.906, 1.067) \\
$\theta_4$ & 0.028  & 0.028  & (0.026, 0.030) \\ 
$\sigma^2$ & 4.087  & 3.818  & (2.018, 7.065) \\ 
\hline 
\end{tabular}\label{table:LV_m_real}
\end{table}
The same model and prior in subsection \ref{subsection:LV} were used. 
As we mentioned in subsection \ref{subsection:u2andm}, we ran the ELW filter 10 times based on $N= 500,000$ particles with $u^2 = 20, 10, 5, 1$ and $0.1^5$, in turn. In this case,  $u^2$ values  smaller than $5$ lead somewhat unstable approximation even with 3,000,000 particles.  
Finally, the state variance was chosen by $u^2= 5$ based on the criterion in subsection \ref{subsection:u2andm}, because it gives stable posterior densities for each ELW run. The other model parameters were chosen as the subsection \ref{subsection:LV}. On average, it took approximately 17 seconds for each run.  

The marginal posterior densities of parameters are given at Figure \ref{fig:LV_m_real}. Posterior summary statistics for the first run are represented at Table \ref{table:LV_m_real}. 
Figure \ref{fig:LVpred_real} contains the scatter plots of the observations and 90\% posterior credible lines for prediction values at 10 future time points when $m=2$ and $u^2=5$. The predicted values of trapped lynxes and hares follow oscillation patterns. The size of prediction interval gets wider as the prediction time gets further ahead and also the predicted value become larger. 

\begin{figure}[!tb]
\centering
\includegraphics[width=17cm, height=10cm]{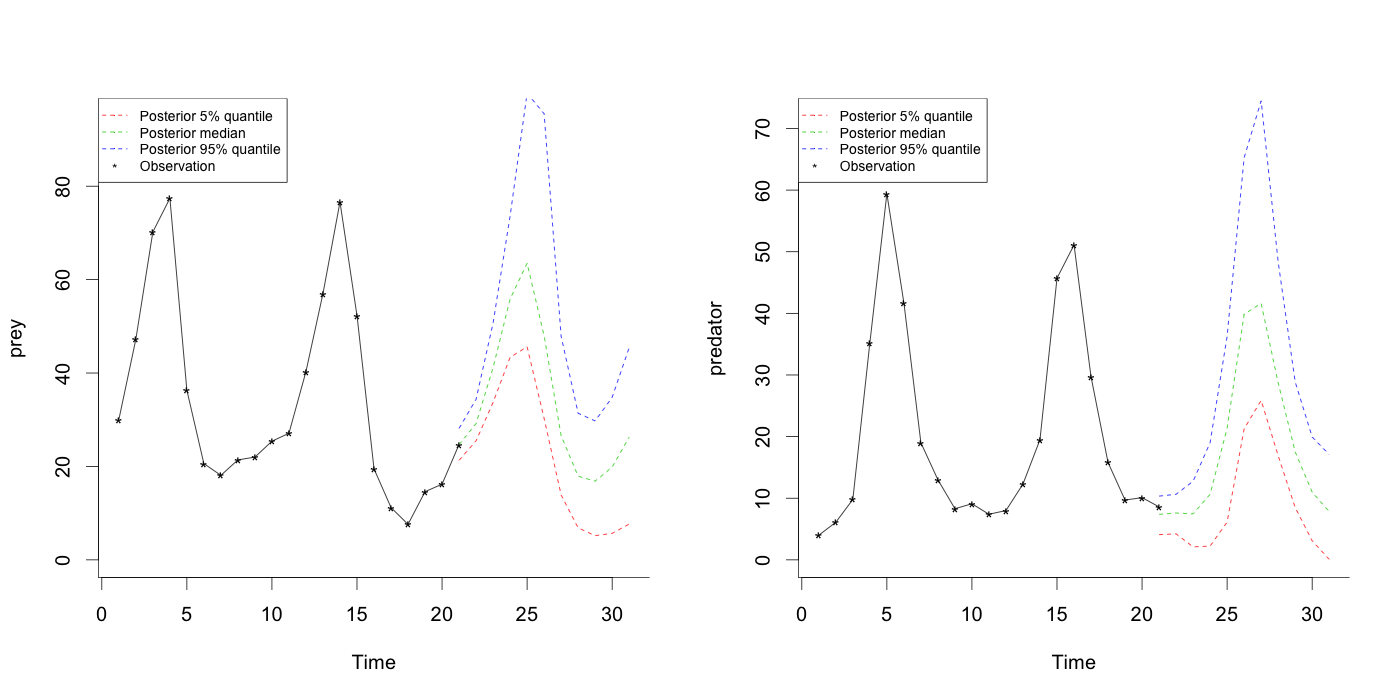}
\caption{Scatter plot of the lynx-hare data and plots of 90\% credible set lines for predictions of 10 time points ahead are drawn when $m = 2$ and $u^2=5$. The upper, lower and middle dotted lines are the 95\% and 5\% quantiles and median of the posterior, respectively. The star-shaped points are the lynx-hare data.}
\label{fig:LVpred_real}
\end{figure}

\section{Discussion}\label{sec:disc}
A lot of biological or physical systems are given by a set of differential equations. To understand these processes, estimation of their parameters is essential. However, especially in Bayesian literature, there is no standard framework to analyze differential equation model. In many cases, the posterior of parameter does not belong a well-known family, so grid sampling or MCMC methods are used to get posterior samples. They usually suffer from heavy computation. 
We propose a general framework to analyze DEM using relaxation via dynamical systems. The dynamic model enables a fast inference for DEM and provides convenient sampling methods.
Among the sampling algorithms for dynamic models, we adopted the ELW filter suggested by Rios and Lopes (2013).  
We argue that our method can be an alternative to the existing inference methods when one needs a fast and reasonable result.
This argument is supported by the example in subsection \ref{subsection:LV}.
Section \ref{section:Conv} guarantees the convergence of the approximated posterior to the true posterior.
However, the theoretical results in this paper does not consider the additional error from the SMC sampling. The proposed method may be improved if a better SMC algorithm is developed.

\section*{Appendix}

The following lemma shows that each $x_i$ given $x_{i-1}, \theta, u^2$ converges to $g(x_{i-1},t_{t-1};\theta)$  in probability as $u^2 \to 0$.
\begin{lemma}\label{inprob}
Consider model \eqref{NDmodel}. Then, for $i =1, \ldots, n$, $x_i$ given $x_{i-1}, \theta$ and $u^2$ converges to $g(x_{i-1}, t_{i-1}; \theta)$  in probability as $u^2 \to 0$.
\end{lemma}

\begin{proof}[Proof of Lemma \ref{inprob}]
Note that $r^T x_i | x_{i-1}, \theta, u^2 \sim N(r^T g(x_{i-1}, t_{i-1} ; \theta) , u^2 \|r\|^2)$ for all $r \in \bbR^p, i = 1, \ldots, n$.
If we denote $\phi_{[Z]}$ as a moment generating function (mgf) of random variable $Z$, then for any $r \in \bbR^p$,
\bean
\phi_{[r^T x_i | x_{i-1}, \theta, u^2]}(z) &=& \exp(r^T g(x_{i-1}, t_{i-1}; \theta) z + \half u^2 \|r\|^2 z^2 )\nonumber \\
&\to& \exp(r^T g(x_{i-1}, t_{i-1}; \theta) z) \label{mgfdelta}
\eean
as $u^2 \to 0$, for $i = 1, \ldots, n$. Note that \eqref{mgfdelta} is mgf of $[ r^T g(x_{i-1}, t_{i-1};\theta) | x_{i-1}, \theta ]$. Since the convergence of mgf implies the convergence of distribution, it implies 
$$[r^T x_i | x_{i-1}, \theta, u^2] \to [ r^T g(x_{i-1}, t_{i-1};\theta) | x_{i-1}, \theta ]$$
for any $r \in \bbR^p$. Hence, by the Cramer-Wold theorem (Billingsley, 1995)\nocite{billingsley1995probability}, it implies that $[ x_i | x_{i-1}, \theta ]$ converges to $g(x_{i-1}, t_{i-1};\theta)$ in distribution, as $u^2 \to 0$.
Note that given $x_{i-1}$ and $\theta$, $g(x_{i-1}, t_{i-1};\theta)$ is a constant. Thus, by Portmanteau theorem (Dudley, 2002)\nocite{dudley2002real}, it implies the convergence in probability. \quad $\square$
\end{proof}

With the continuity condition of $f(x,t ; \theta)$ in $x$, Lemma \ref{inprob} can be extended to the joint convergence in probability using the mathematical induction. Lemma \ref{x1xn} describes the result.

\begin{lemma}\label{x1xn}
Consider model \eqref{NDmodel}. Suppose $f( x, t ; \theta)$ is continuous in $x$. Then, $[x_1, \ldots, x_n \mid x_0, \theta, u^2]$ converges to $(g(x_0, t_0;\theta), \ldots, g^{n}(x_0, t_{n-1};\theta) )$ in probability as $u^2 \to 0$. 
\end{lemma}

\begin{proof}[Proof of Lemma \ref{x1xn}]
Let $X = (x_1, \ldots, x_n)$ and $\bar{X} = (g(x_0, t_0;\theta), \ldots, g^{n}(x_0, t_{n-1};\theta) )$ where
\bean\label{induc}
x_{i}^m = g^{i}( x_0, t_{i-1} ; \theta), ~~ i = 1, \ldots, n
\eean
by the relation \eqref{RK4} where $g^{i}(x_0, t_{i} ;\theta) = g( g^{i-1}(x_0, t_{i-1};\theta), t_i ; \theta)$ is defined recursively.
We want to show
$$\lim_{u^2 \to 0} P\Big( \| X - \bar{X} \| \ge \epsilon | x_0, \theta, u^2 \Big) = 0$$
for given $\epsilon > 0$. It suffices to prove
\bean\label{jointprob}
\lim_{u^2 \to 0} P\Big( \| x_i - g^{i}(x_0, t_{i-1};\theta) \| \ge \frac{\epsilon}{n} | x_0, \theta , u^2 \Big) = 0
\eean
for given $\epsilon > 0$ and $i = 1,\ldots, n$. We use the mathematical induction.

When $i  = 1$, we can check
\bea
\lim_{u^2 \to 0} P \Big( \| x_1 - g(x_0, t_0 ;\theta) \| \ge \frac{\epsilon}{n} | x_0, \theta, u^2 \Big) = 0
\eea
by Lemma \ref{inprob}.
Suppose \eqref{jointprob} holds for $i = k$. 
Note
\bean
&& P( \| x_{k+1} - g^{k+1}(x_0, t_k ; \theta) \| \ge \frac{\epsilon}{n} | x_0, \theta, u^2) \nonumber\\
&\le& P( \| x_{k+1} - g(x_k, t_k ; \theta) \| \ge \frac{\epsilon}{2n} | x_0, \theta, u^2)\label{T1}  \\
&+& P( \| g(x_k, t_k ; \theta) - g( g^{k}(x_0, t_{k-1} ; \theta), t_k ; \theta) \| \ge \frac{\epsilon}{2n} | x_0, \theta, u^2).\label{T2}
\eean
By assumption, $g(x,t | \theta)$ is continuous in $x$. Thus, \eqref{T2} converges to 0 as $u^2 \to 0$ because \eqref{jointprob} holds for $i=k$. 
Also note that \eqref{T1} is 
\bea
E_{x_2 | x_0, \theta, u^2} \ldots E_{x_k | x_{k-1}, \theta, u^2} \Big[ P( \| x_{k+1} - g(x_k, t_k ; \theta) \| \ge \frac{\epsilon}{2n} | x_k, \theta, u^2) \Big].
\eea
Since $P( \| x_{k+1} - g(x_k, t_k ; \theta) \| \ge \epsilon/(2n) | x_k, \theta, u^2) \le 1$ and Lemma \ref{inprob}, \eqref{T1} converges to 0 as $u^2 \to 0$ by the bounded convergence theorem. ~~~ $\square$
\end{proof}

\begin{proof}[Proof of Theorem \ref{convtoDE}]
Note that we need to prove
\bean
\int L(\Lambda) \pi (dx_1, \ldots, dx_n | x_0, \theta, u^2)\pi(x_0, \theta, \lambda) &\to& L^*(x_0, \theta, \lambda) \pi(x_0, \theta, \lambda), \label{num} \\
\int \int L(\Lambda) \pi (dx_1, \ldots, dx_n | x_0, \theta, u^2) \pi(dx_0, d\theta, d\lambda) &\to& \int L^*(x_0, \theta, \lambda) \pi(dx_0, d\theta, d\lambda)\label{denom} \hspace{1cm}
\eean
as $u^2 \to 0$ where $\Lambda = (x_1, \ldots,  x_n, \theta, \lambda)$. 

To show \eqref{num}, we only need to prove
\bea
\int L(\Lambda) \pi(dx_1, \ldots , dx_n | x_0, \theta, u^2) \to L^*(x_0, \theta, \lambda)
\eea
as $u^2 \to 0$. Since $L(\Lambda) = \lambda^{{np}/{2}} \exp({-\frac{\lambda}{2}\sum_{i=1}^n \|y_i - x_i \|^2})$, it suffices to prove
\bean\label{num2}
\int e^{-\frac{\lambda}{2}\sum_{i=1}^n \|y_i - x_i \|^2} \pi(dx_1, \ldots , dx_n | x_0, \theta, u^2) \to e^{-\frac{\lambda}{2} \sum_{i=1}^n \|y_i - g^{i-1}(x_0, t_{i-1} ; \theta) \|^2}.
\eean
By Lemma \ref{x1xn}, we have 
$$[x_1, \ldots, x_n | x_0, \theta, u^2] \to [g(x_0, t_1 ; \theta) , \ldots, g^{n-1}(x_0, t_{n-1} ; \theta) | x_0, \theta]$$
as $u^2 \to 0$. Note that the right hand side of \eqref{num2} is the expectation of $\exp({-{\lambda}/{2}\cdot\sum_{i=1}^n \|y_i - x_i \|^2})$ with respect to $[g(x_0, t_1 ; \theta) , \ldots, g^{n-1}(x_0, t_{n-1} ; \theta) | x_0, \theta]$. Also note that $\exp({-{\lambda}/{2}\cdot\sum_{i=1}^n \|y_i - x_i \|^2})$ is bounded by 1 and is continuous in $x_1, \ldots, x_n$. Thus, the Portmanteau theorem implies \eqref{num}.

Since we have proved \eqref{num}, it suffices for \eqref{denom} to show that $\int L(\Lambda) \pi(dx_2, \ldots , dx_n | x_0, \theta, u^2)$ is dominated by an integrable random variable. It is easy to check because
\bea
\int L(\Lambda) \pi(dx_2, \ldots , dx_n | x_0, \theta, u^2) &\le& (\lambda)^{\frac{np}{2}}
\eea
and $(\lambda)^{{np}/{2}}$ is integrable with respect to $\pi(x_0, \theta, \lambda)$. The dominated convergence theorem gives the desired result. ~~~ $\blacksquare$
\end{proof}

\begin{proof}[Proof of Theorem \ref{convtoTrue}]
Denote the likelihood of approximated $x$ with the number of segments $m$ as $L_m(x_0, \theta, \lambda)$, and let $L_{\text{true}}(x_0, \theta, \lambda)$ be the likelihood of true $x$. 
We should prove that
\bea
\pi_m (x_0,\theta,\lambda | {\bf y}_n) = \frac{L_m(x_0,\theta,\lambda) \pi(x_0,\theta,\lambda)}{\int L_m(x_0,\theta,\lambda) \pi(dx_0,d\theta,d\lambda)}
\eea
converges to
\bea
\pi_{\text{true}} (x_0,\theta,\lambda | {\bf y}_n) = \frac{L_{\text{true}}(x_0,\theta,\lambda) \pi(x_0,\theta,\lambda)}{\int L_{\text{true}}(x_0,\theta,\lambda) \pi(dx_0,d\theta,d\lambda)}
\eea
for any $x_0, \theta$ and $\lambda$.
It is well known that if $f(x,t ; \theta)$ satisfies Lipschitz condition in $x$, then Runge-Kutta method converges to the true solution, i.e.
\bean\label{RKconv}
x^m_i(x_0, \theta) \to x_i(x_0, \theta) ~\text{ for all $x_0$ and $\theta$ as $m \to \infty$.}
\eean
See Cartwright and Piro (1992)\nocite{Cartwright92} for the proof. 
The convergence \eqref{RKconv} implies that $L_m(x_0,\theta,\lambda)$ converges to $L_{\text{true}}(x_0,\theta,\lambda)$ for all $x_0,\theta$ and $\lambda$ because an exponential function is continuous. It implies the convergence of numerator part.

For the denominator part, recall that
$$L_m(x_0, \theta,\lambda) \le (\lambda)^{\frac{np}{2}}$$
and $(\lambda)^{{np}/{2}}$ is integrable with respect to $\pi(x_0, \theta, \lambda)$. Again, the dominated convergence theorem gives the desired result. ~~~ $\blacksquare$
\end{proof}

\begin{proof}[Proof of Theorem \ref{errorrate}]
At first, we want to show that under $A1-A3$, $| ng_n(x_0) - ng_n^m(x_0)| = O(n (h/m)^K)$ for sufficiently large $n$. 
Since we assume the Lipschitz continuity of $f$, the ODE has a unique solution with initial condition $x(t_1) = x_0$. Assumptions A1 and A3 implies  
$$ \sup_{x, t} \| \frac{d^K}{dt^K} f(x,t;\theta) \| =: B < \infty$$
for some constants $B >0$. 
The  local errors of the $K$th order numerical method are given by
$$ \| x(t_i) - x(t_{i-1}) - h \phi(x_{i-1} , t_{i-1} ; \theta) \|  \le B' h^{K+1}, ~ i=1,\ldots, n$$
for some $B'>0$, which 
depends only on $\sup_{ t} \| d^K f(x,t;\theta)/ (dt^K) \| \le B$ (Palais and Palais, 2009).\nocite{Palais09}
Thus, the  local errors are uniformly bounded. It implies the global errors uniformly bounded by
$$\|x_i - x^h_i\| \le C h^K $$
for some constant $C>0$. Thus, 
\begin{eqnarray*}
| ng_n(x_0) - ng_n^m(x_0)| &=& \big| \sum_{i=1}^n \|y_i - x_i \|^2 - \sum_{i=1}^n\|y_i - x_i^m\|^2 \big|  \\
& = & \sum_{i=1}^n \big(  \|y_i - x_i \| +  \|y_i - x_i^m \| \big) \big|  \|y_i - x_i \| -  \|y_i - x_i^m \| \big|  \\
&\le& \sum_{i=1}^n \big( 2 \|y_i - x_i \| + \|x_i -x_i^m\| \big) \| x_i-x_i^m\|  \\
&\le& \sum_{i=1}^n \big( 2C_y + 2 C_x + \|x_i -x_i^m\| \big) \|x_i -x_i^m\|  \\
&\le&\sum_{i=1}^n \left(  2C_y + 2 C_x + C \Big(\frac{h}{m}\Big)^K \right) C \Big(\frac{h}{m}\Big)^K   \\
&\asymp& n \Big(\frac{h}{m} \Big)^K,
\end{eqnarray*}
where $\sup_{t \in [T_0,T_1]}\|y(t)\| < C_y < \infty$, $\sup_{t \in [T_0,T_1]}\|x(t)\| < C_x <\infty$ for sufficiently large $n$.   

By the above inequality, for fixed $x_0 \in \mathbb{R}^p, \lambda>0$,
\bea
e^{-\frac{\lambda}{2}ng_n^m(x_0)} &=& e^{-\frac{\lambda}{2}[ng_n(x_0)+ng_n^m(x_0)-ng_n(x_0)]}\\
&=&  e^{-\frac{\lambda}{2}ng_n(x_0)} \times e^{-\frac{\lambda}{2}[ng_n^m(x_0)-ng_n(x_0)]} \\
&=& e^{-\frac{\lambda}{2}ng_n(x_0)} \times e^{-\frac{\lambda}{2}O(n (\frac{h}{m})^K)} \\
&=& e^{-\frac{\lambda}{2}ng_n(x_0)} \times \left( 1+ O\Big(n \Big(\frac{h}{m}\Big)^K\Big) \right)
\eea
because $e^{x} = 1+ O(x)$ for sufficiently small $x$. 
It implies 
\begin{eqnarray*}
\pi_m( x_0, \theta, \lambda \mid {\bf y}_n) &\propto&   L_m(\theta, \lambda, x_0) \pi(\theta,\lambda,x_0)   \\
&=&   L^*(\theta, \lambda, x_0) \pi(\theta,\lambda,x_0) \times \left( 1+ O\Big(n \Big(\frac{h}{m}\Big)^K \Big) \right)  \\
&\propto& \pi(x_0, \theta,\lambda\mid {\bf y}_n) \times \left( 1+ O\Big(n \Big(\frac{h}{m}\Big)^K \Big) \right) 
\end{eqnarray*}
for sufficiently large $n$. If $\alpha > (1+R)/K$, then we have $n(h/m)^K \le n^{-R}$. ~~~ $\blacksquare$
\end{proof}

\newpage
\bibliographystyle{plain}
\bibliography{dynamic_paper}

\end{document}